\documentclass{imsart}

\usepackage{amsthm}
\usepackage{amsmath}
\usepackage{multirow}
\usepackage{graphicx}
\usepackage{subfigure}

 \startlocaldefs
\newtheorem{lemma}{Lemma}
\newtheorem{thm}{Theorem}

\endlocaldefs
 
\begin{document}
\begin{frontmatter}
\title{A nonparametric copula density estimator incorporating information on bivariate marginals}
\runtitle{Nonparametric copula density estimator}

\author{\fnms{Yu-Hsiang} \snm{Cheng}\ead[label=e1]{96354501@nccu.edu.tw}}
\address{\printead{e1}}
\affiliation{Academia Sinica, Taiwan R.O.C.}
 \and
 \author{\fnms{Tzee-Ming} \snm{Huang}\corref{}\ead[label=e2]{tmhuang@nccu.edu.tw}\thanksref{t1}}
\thankstext{t1}{Corresponding author.} 
\address{\printead{e2}}
\affiliation{National Chengchi University, Taiwan R.O.C.}
\runauthor{Cheng and Huang}

\end{frontmatter}
\section{Introduction}
Consider the problem of estimating  copula density when the bivariate marginals are known. Let $c$ be a copula density. Let $W_h$ be  the space of tensor product linear B-splines on $[0,1]^d$ with equally space knots, where $h = (h_1, \ldots, h_d)$ and $h_i$ is the distance between two adjacent knots for the $i$-th dimension. Let $B_1$, $\ldots$, $B_k$ denote the tensor product B-spline basis functions for $W_h$. Then, we consider $c$ is approximated by
\[ 
   \alpha_1 B_1 + \cdots + \alpha_k B_k ,
\]
where $\alpha_1, \ldots, \alpha_k$ are coefficients. First, we state some notations and definitions as follows.
\begin{itemize}
\item $W = L^2([0, 1]^d)$.
\item $S$: $\{(f_1, \ldots, f_d): f_1, \ldots, f_d \in L^2([0, 1]) \mbox{ and } \int_0^1 f_1(x)dx = \cdots = \int_0^1 f_d(x)dx\}.$
\item  $A: W \rightarrow S$ is a linear operator such that for $f$  in $W$, $Af = (f_1, \ldots, f_d)$, where \[f_i(x_i)=\int_0^1 \cdots \int_0^1 f(x_1, \ldots, x_d)dx_1 \cdots dx_{i-1}dx_{i+1} \cdots dx_d.\]
\item $H_{ij}: L^2([0, 1]^d) \rightarrow L^2([0, 1]^2)$ to be the linear mapping such that for $f \in W$ and  for $1 \leq i, j \leq d$,
$(H_{ij}f)(x_i, x_j)$  is given by 
\[\int_0^1 \cdots \int_0^1 f(x_1, \ldots, x_d)dx_1 \cdots dx_{i-1}dx_{i+1} \cdots dx_{j-1}dx_{j+1}\cdots dx_d.\]
\item $M = \Big(\int_{[0, 1]^d} c(u)B_1(u)du, \ldots, \int_{[0, 1]^d} c(u)B_k(u)du \Big)^T$
\end{itemize}   
\section{Methodology and main results}
Suppose that we observe data  $X_1, \ldots, X_n$ with copula density $c$, where $X_i = (X_{i,1}, \ldots, X_{i,d})$. For $i = 1, \ldots, k$, let 
\[
 \hat{M}_i = \frac{1}{n}\sum_{j = 1}^n B_i \big(\hat{F}_1(X_{j,1}), \ldots, \hat{F}_d(X_{j,d})\big)
\] 
be a moment estimator for $M_i$, where $ \hat{F}_j(x) = \frac{1}{n} \sum_{i = 1}^n I(X_{i,j} \leq x)$ is the empirical CDF of data $X_{1,j}, \ldots, X_{n,j}$. Let
\[
  P = \left(
\begin{array}{ccc}
\int_{[0, 1]^d} B_1(u)B_1(u)du & \cdots & \int_{[0, 1]^d} B_1(u)B_k(u)du \\
\vdots & \ddots & \vdots \\
\int_{[0, 1]^d} B_k(u)B_1(u)du & \cdots & \int_{[0, 1]^d} B_k(u)B_k(u)du 
\end{array} \right)\]
and 
\[
 Pen(c, \alpha) = \sum_{1 \leq i, j \leq d} \int_0^1 \int_0^1 \Big|(H_{ij}c)(u_i, u_j)-(H_{ij}\sum_{t = 1}^k \alpha_t B_t)(u_i, u_j)\Big|^2 du_i du_j, \]
where $\alpha = (\alpha_1, \ldots, \alpha_k)^T$ is the vector of B-spline coefficients.
We estimate the copula density $c$ using $\hat{c} = \hat{\alpha}_1 B_1 + \cdots + \hat{\alpha}_k B_k$, where $\hat{\alpha} = (\hat{\alpha}_1, \ldots, \hat{\alpha}_k)$ is the minimizer of
\[
 \beta(P \alpha - \hat{M})^T(P \alpha - \hat{M}) + \lambda Pen(c, \alpha)
\]
under the constraints that $\alpha^T B$ is nonnegative and the marginals of $\alpha^T B$ are uniform density on $[0,1]$, where $B= (B_1, \ldots, B_k)^T$ and $\beta =1/(\prod_{i=1}^d h_i)$. Now, we show that $\hat{c}$ is consistent for $c$ under some mild conditions. 
Theorem 1 gives the approximation error of B-splines under a linear constraint and a nonnegativity consraint.
\begin{thm} \label{thm1}
Suppose that $\eta \in S$ and $\{f \in W_h : Af = \eta\} \neq \emptyset$. Let $V$ denote the set $\{f \in W: f \geq 0 \}$ and $V^{(0)}$ be the interior of $V$. Suppose that $g \in \{ f \in V^{(0)}: Af = \eta \}$ and $\varepsilon$ is a positive number such that $B(g, \varepsilon) \stackrel{\mbox{def}}{=} \{ f \in W: \| f - g \| < \varepsilon \} \subset V^{(0)}$ and 
\[
 2d \|\bar{g}_w - g\| < \frac{\varepsilon}{2}
\]
for some $\bar{g}_w \in W_h$. Then for $f \in V \cap \{f \in W: Af = \eta\}$ and $\bar{f}_w \in W_h$, there exists $f_w \in V \cap\{f \in W_h: Af = \eta\}$ such that
\[\| f_w-f\| \leq 2d \Big(1+\frac{2}{\varepsilon}(\|f\|+\|g\|+\varepsilon)\Big) \|\bar{f}_w-f\|. \]
%where $\varepsilon$ depends on $V^{(0)} \cap \{ f \in W: Af = \eta \}$ and $g$.
\end{thm}
Using Theorem 1, we can establish the consistency of $\hat{c}$, and the result is given in Theorem 2. 
\begin{thm} \label{thm2}
Let $k = \prod_{i=1}^d \big((1/h_i)+1\big)$ denote the number of tensor basis functions. Suppose that $\lim \limits_{n \to \infty} \max(h_1, \ldots, h_d) = 0$ and $\lim \limits_{n \to \infty}(1/h_{\min})^{2+d}k/n = 0$, where $h_{\min} = \min(h_1, \ldots, h_d)$ . Then, $\|\hat{c}-c\| \rightarrow 0$ in probability.
\end{thm}
\section{Proofs}
We will provide the proofs of Theorems 1 and 2 in this section.
\subsection{Proof of Theorem 1}
The  proof of Theorem 1 is based on Lemma 1, which is stated and proved below.
\begin{lemma} \label{lmm1}
Suppose that $\eta \in S$ and $\{f \in W_h: Af = \eta\} \neq \emptyset$. Then for $f \in \{f \in W: Af = \eta\}$ and $\bar{f}_w \in W_h$, there exists $f_w \in \{f \in W_h: Af = \eta\}$ such that
\[\| f_w-f\| \leq 2d\|\bar{f}_w-f\|,\]
where $\| \cdot\|$ denotes the $L^2$ norm.
\end{lemma}
\begin{proof}
First, we will prove Lemma \ref{lmm1} when $\eta=(0, \ldots, 0)$. For any  $\bar{f}_w \in W_h$, let $(f_1, \ldots, f_d) = A\bar{f}_w$ and $\mu = \int_0^1 f_1(x)dx$. Let $g(x_1, \ldots, x_d) = \sum_{i=1}^d f_i(x_i)$ and $f^* = \bar{f}_w - g$, then we have
\[ Af^* = A\bar{f}_w-Ag = -\mu (d-1)(1, \ldots, 1).\]
Let $e_1$ denote the constant function $1$ on $[0, 1]^d$, then $e_1 \in W_h$. Take $f_w = f^*+\mu (d-1)e_1$, then $Af_w = (0, \ldots, 0)$ and 
\begin{eqnarray*}
\| f_w-f\| & \leq & \| f_w-\bar{f}_w \| + \| \bar{f}_w-f \| \\
           & \leq & \| \mu (d-1) e_1 - g \| + \| \bar{f}_w-f \| \\
           & \leq & |\mu| (d-1) \| e_1  \|+ d \| \bar{f}_w-f \| + \| \bar{f}_w-f \| \\
           & \leq & 2d\| \bar{f}_w-f \| .
\end{eqnarray*}
Here we have used the fact that $\mu^2 \leq \|f_i\|^2 \leq \|\bar{f}_w-f\|^2$ for $i=1$, $\ldots$, $d$. 

Next, we will prove  Lemma \ref{lmm1}  for a general $\eta$. From the assumption that $\{ f \in W_h: Af=\eta \} \neq \emptyset$, there exists a function $\tilde{f}$  in $W_h$ such that $A\tilde{f} = \eta$. Suppose that  $f \in \{ f \in W: Af=\eta \}$ and $\bar{f}_w \in W_h$. Then $A(f-\tilde{f}) = (0, \ldots, 0)$. Apply Lemma \ref{lmm1} with $\eta$, $f$ and $\bar{f}_w$  replaced by $(0, \ldots, 0)$, $f-\tilde{f}$ and $\bar{f}_w-\tilde{f}$ respectively, then there exists $g_w \in W_h$ such that $Ag_w = (0, \ldots, 0)$ and
\[  
 \|g_w-(f-\tilde{f})\| \leq 2d\|(\bar{f}_w-\tilde{f})-(f-\tilde{f}) \|.
\]
Take $f_w = g_w+\tilde{f}$, then $f_w \in W_h$, $Af_w = \eta$ and the above equation becomes
\[
  \| f_w - f \| \leq 2d\| \bar{f}_w- f \|.
\]
The proof of Lemma \ref{lmm1} is complete.
\end{proof}
Now, we will prove Theorem 1.
\begin{proof}
The proof for Theorem \ref{thm1} is adapted from the proof for Lemma 2.4 in Wong \cite{wong1984}. 
Suppose that the assumptions in Theorem \ref{thm1} hold, $f \in V \cap \{f \in W: Af = \eta\}$ and $\bar{f}_w \in W_h$.  Then by Lemma \ref{lmm1}, there exist $\tilde{g}_w$, $\tilde{f}_w \in \{f \in W_h: Af = \eta\}$ such that
\[
  \| \tilde{g}_w - g\| \leq 2d\|\bar{g}_w - g\| < \frac{\varepsilon}{2} 
\]
and 
 \[
 \|  \tilde{f}_w - f\| \leq 2d\|\ \bar{f}_w - f|.
 \]
Note that $\tilde{g}_w  \in \{f \in W_h: Af = \eta\} \cap V^{(0)}$. Let $f_\tau = \tau \tilde{f}_w + (1 -  \tau)\tilde{g}_w$ and $ \tau^* = \sup \{ \tau \in [0, 1]: f_ \tau \in V \}$, then we will show that Theorem \ref{thm1} holds with $f_w= f_{ \tau^*}$. 

Take $\varepsilon_1 = \varepsilon/2$, then $B(\tilde{g}_w, \varepsilon_1) \subset B(g, \varepsilon) \subset V^{(0)}$.  For
\[ 0 \leq \tau \leq \frac{\varepsilon_1}{\varepsilon_1 + \| \tilde{f}_w - f\|}, \] we have 
$\|\tau (\tilde{f}_w - f)/(1 - \tau)\| \leq \varepsilon_1$, so
\[
 \frac{\tau}{1 - \tau}(\tilde{f}_w - f) + \tilde{g}_w \in B(\tilde{g}_w , \varepsilon_1)  \subset V,
\]
which gives
\begin{eqnarray*}
    f_\tau &=&    \tau \tilde{f}_w + (1 - \tau)\tilde{g}_w \\
                  &=&    \tau f + (1 - \tau) \Big[\frac{\tau}{1 - \tau}(\tilde{f}_w - f) + \tilde{g}_w \Big] \in V.
\end{eqnarray*}
Therefore, we have 
\[
  1 - \tau^* \leq 1- \frac{\varepsilon_1}{\varepsilon_1 + \|\tilde{f}_w - f\|} =\frac{\|\tilde{f}_w - f\|}{\varepsilon_1 + \| \tilde{f}_w- f\|} \leq 
 \frac{\| \tilde{f}_w - f\|}{\varepsilon_1}
\]
and 
\begin{eqnarray*}
\|f_w - f\| & = & \|\tau^* \tilde{f}_w + (1 - \tau^*) \tilde{g}_w - f\| \\
& \leq & \tau^* \| \tilde{f}_w - f\| + (1 - \tau^*) \|\tilde{g}_w- f\|\\
& \leq & \|  \tilde{f}_w - f\| \Big( 1+\frac{2}{\varepsilon} (\|f\| + \|g\| + \varepsilon) \Big)\\
& \leq & 2d \Big(1+\frac{2}{\varepsilon}(\|f\|+\|g\|+\varepsilon)\Big) \|\bar{f}_w-f\|.\\
\end{eqnarray*}
\end{proof}

\subsection{Proof of Theorem 2}
Before we provide the proof of Theorem 2, the Lemma2 and its proof are stated as follows.
\begin{lemma}
For $t = 1, \ldots, k$, $ \displaystyle \hat{M}_t = \frac{1}{n}\sum_{i = 1}^n B_t \big(\hat{F}_1(X_{i,1}), \ldots, \hat{F}_d(X_{i,d})\big)$. Let $h_{\min}=\min(h_1, \ldots, h_d)$. Then $E(|\hat{M}_t - M_t|^2) = O\big( \frac{1}{n} \big) \Big(1+ \left( \frac{d}{h_{\min}} \right)^2\Big)$ .
\end{lemma}
\begin{proof}
For simplicity, We first define some notations. For  $1 \leq \ell \leq d$ and $1 \leq i \leq n$, let $\zeta_{i, \ell} = \hat{F}_{\ell}(X_{i, \ell}) - F_{\ell}(X_{i, \ell})$, where $F_{\ell}$ is the CDF for $X_{i,\ell}$,
and
\begin{eqnarray*}
\varphi_1 & = & B_t \big( F_1(X_{i,1})+\zeta_{i,1}, \ldots, F_d(X_{i,d})+\zeta_{i,d}  \big) \\
&\vdots & \\
\varphi_{\ell} & = & B_t \big(F_1(X_{i,1}), \ldots, F_{\ell-1}(X_{i,\ell-1}), F_{\ell} (X_{i,\ell})+\zeta_{i,\ell}, \ldots, F_d(X_{i,d})+\zeta_{i,d}  \big)\\
&\vdots & \\
\varphi_d & = & B_t \big(F_1(X_{i,1}), \ldots, F_{d-1}(X_{i,d-1}), F_d(X_{i,d})+\zeta_{i,d} \big)\\
\varphi_{d+1} & = & B_t \big(F_1(X_{i,1}), \ldots, F_d(X_{i,d}) \big).
\end{eqnarray*}
Since $B_t(u_1, \ldots, u_d) = \phi_{t,1}(u_1)\cdots\phi_{t,d}(u_d)$, where for $1 \leq j \leq d$, $|\phi_{t,j}| \leq 1$ and 
\[ 
         | \phi_{t,j}(x_1) - \phi_{t,j}(x_2)|    \leq \frac{|x_1-x_2|}{h_{\min}}  \mbox{ for } x_1, x_2 \in [0,1],
\]
we have
\[\varphi_{\ell} - \varphi_{\ell+1} = \Big[\phi_{t, \ell}\big(F_{\ell} (X_{i, \ell})+\zeta_{i,\ell}\big) - \phi_{t, \ell}\big(F_{\ell} (X_{i, \ell})\big) \Big] \prod_{s = 1}^{\ell-1} \phi_{t, s}\big( F_s(X_{i,s})\big) \prod_{s = \ell+1}^d \phi_{t, s} \big( F_s(X_{i,s})+ \zeta_{i,s} \big),\]
and
\[|\varphi_{\ell} - \varphi_{\ell+1}| \leq |\phi_{t, \ell}\big(F_{\ell} (X_{i,\ell})+\zeta_{i,\ell} \big)-\phi_{t, \ell}\big(F_{\ell} (X_{i,\ell})\big)| \leq \frac{1}{h_{\min}}|\zeta_{i,\ell}|.\]
Therefore,
\begin{eqnarray*}
& & |B_t\big(\hat{F}_1(X_{i,1}), \ldots, \hat{F}_d(X_{i,d})\big)-B_t\big(F_1(X_{i,1}), \ldots, F_d(X_{i,d})\big)|\\
&=& |(\varphi_1 - \varphi_2) + (\varphi_2 - \varphi_3) + \cdots + (\varphi_d - \varphi_{d+1})|\\
& \leq & \frac{1}{h_{\min}} \sum_{\ell=1}^d |\zeta_{i,\ell}|.
\end{eqnarray*}
In addition,
\begin{eqnarray*}
|M_t - \hat{M}_t| &=& \Big| 
\int_{[0, 1]^d}c(u)B_t(u)du - \frac{1}{n}\sum_{i= 1}^n B_t \big(\hat{F}_1(X_{i,1}), \ldots, \hat{F}_d(X_{i,d})\big)\Big|\\
& \leq & \bigg| \underbrace{\frac{1}{n}\sum_{i= 1}^n \Big[B_t \big(\hat{F}_1(X_{i,1}), \ldots, \hat{F}_d(X_{i,d})\big) - B_t \big(F_1(X_{i,1}), \ldots, F_d(X_{i,d})\big)\Big] }_{I} \bigg|\\
&& + \Big| \underbrace{\frac{1}{n} \sum_{i = 1}^n B_t \big(F_1(X_{i,1}), \ldots, F_d(X_{i,d})\big) - E\big(B_t(F_1(X_{i,1}), \ldots, F_d(X_{i,d})\big) }_{II} \Big|.
\end{eqnarray*}
Since $|I| \leq \frac{1}{nh_{\min}} \sum_{i =1}^n \sum_{\ell = 1}^d |\zeta_{i,\ell}|$, we have
\[EI^2 \leq \Big(\frac{1}{nd}\Big)\Big(\frac{d}{h_{\min}}\Big)^2 \sum_{i =1}^n \sum_{\ell = 1}^d E\zeta_{i,\ell}^2 = \Big(\frac{1}{nd}\Big)\Big(\frac{d}{h_{\min}}\Big)^2 \sum_{j =1}^n \sum_{\ell = 1}^d \frac{2n+2}{12n^2}=\Big(\frac{d}{h_{\min}}\Big)^2 O\Big(\frac{1}{n}\Big).\]
In addition,
\[EII^2 \leq \frac{1}{n} = O\Big(\frac{1}{n}\Big),\]
so
\[E(| M_t - \hat{M}_t |^2) = O\Big(\frac{1}{n}\Big) \Big(1+\Big(\frac{d}{h_{\min}}\Big)^2\Big).\]
\end{proof} 
Now, we will prove Theorem 2.
\begin{proof}
Suppose that there exists a $\alpha^*$ such that $\|(\alpha^*)^TB-c\| \leq \Delta_1$, where $(\alpha^*)^T B$ satisfies the constraints for a copula density.
Then,
\begin{eqnarray}
\|\hat{\alpha}^TB-c\|^2 & \leq & 2\|( \hat{\alpha}-\alpha^*)^T B\|^2 + 2\|(\alpha^*)^TB-c\|^2 \nonumber \\
& \leq & 2( \hat{\alpha}-\alpha^*)^T P ( \hat{\alpha}-\alpha^*) +2 \Delta_1^2 \nonumber \\
& \leq & 2( \hat{\alpha}-\alpha^*)^T P^T P^{-1}P ( \hat{\alpha}-\alpha^*) +2 \Delta_1^2 \nonumber\\
& \leq & \frac{2}{\min(eigen(P))}( \hat{\alpha}-\alpha^*)^T P^T P ( \hat{\alpha}-\alpha^*) +2 \Delta_1^2 \nonumber\\
& \leq & \Big(\frac{2}{\beta\min(eigen(P))}\Big) \Big(\beta ( \hat{\alpha}-\alpha^*)^T P^T P ( \hat{\alpha}-\alpha^*)\Big)+2 \Delta_1^2. \label{eq:11}
\end{eqnarray}
Here the $eigen(A)$ denotes the eigenvalues of a matrix A.
Let 
\begin{eqnarray*}
 I_0 & = & \beta (P \hat{\alpha}-\hat{M})^T(P \hat{\alpha} -\hat{M})+\lambda Pen(c, \hat{\alpha})\\ 
&  &- [\beta (P \alpha^* -\hat{M})^T(P \alpha^* -\hat{M})+\lambda Pen(c, \alpha^*)] \\
 & & - [\beta (P \hat{\alpha} -M)^T(P \hat{\alpha} -M)+\lambda Pen(c, \hat{\alpha}) ]\\
 &  & + \beta (P \alpha^* -M)^T(P \alpha^* -M)+\lambda Pen(c, \alpha^*). 
\end{eqnarray*}
Since 
\begin{eqnarray*}
&&\beta (P \hat{\alpha}-\hat{M})^T(P \hat{\alpha}-\hat{M})+\lambda Pen(c, \hat{\alpha})\\ 
& \leq & \beta (P \alpha^*-\hat{M})^T(P \alpha^*-\hat{M})+\lambda Pen(c, \alpha^*),
\end{eqnarray*}
we have
\begin{eqnarray*}
&& I_0 + \beta (P \hat{\alpha} -M)^T(P \hat{\alpha} -M)+\lambda Pen(c, \hat{\alpha})\\ 
& \leq & \beta (P \alpha^* -M)^T(P \alpha^*-M)+\lambda Pen(c, \alpha^*) .
\end{eqnarray*}
Therefore,
\begin{eqnarray}
&& \beta (\hat{\alpha} - \alpha^*)^TP^TP(\hat{\alpha} - \alpha^*) \nonumber\\
& & \leq 2 \beta (P\hat{\alpha} - M)^T(P\hat{\alpha} - M) + 2 \beta (P\alpha^* - M)^T(P\alpha^* - M) \nonumber \\
&& \leq 4 \beta (P\alpha^* - M)^T(P\alpha^* - M) + 2 \lambda Pen(c, \alpha^*) -2 I_0.  \label{eq:1}
\end{eqnarray}
Note that \[I_0 = 2 \beta (\hat{M}-M)^T P \alpha^* - 2 \beta (\hat{M}-M)^T P \hat{\alpha},\]
so
\begin{eqnarray}
-2I_0 & \leq & 4|\beta(\hat{M}-M)^T P (\hat{\alpha}-\alpha^*)| \nonumber \\
& \leq & 4  \sqrt{\beta (\hat{M}-M)^T(\hat{M}-M)}\sqrt{\beta \big(P(\hat{\alpha}-\alpha^*)\big)^T \big(P(\hat{\alpha}-\alpha^*)\big)} \label{eq:2}
\end{eqnarray}
Let
\begin{eqnarray*} 
 \varepsilon_1 &=& 4 \beta (P\alpha^* - M)^T(P\alpha^* - M) + 2 \lambda Pen(c, \alpha^*), \\
 \varepsilon_2 &=&  \sqrt{\beta (\hat{M}-M)^T(\hat{M}-M)}, 
\end{eqnarray*} 
and
\[
 U= \sqrt{\beta \big(P(\hat{\alpha}-\alpha^*)\big)^T \big(P(\hat{\alpha}-\alpha^*)\big)},
\]
then it follows from (\ref{eq:1}) and (\ref{eq:2}) that 
\[
 U^2 \leq \varepsilon_1 + 4 \varepsilon_2 U,
\]
so 
\begin{equation} \label{eq:3}
 |U| \leq 2\varepsilon_1 + \sqrt{\varepsilon_1+4\varepsilon_2^2}.
\end{equation}
 
To control $\varepsilon_1$,  let $c_i = \int_{[0, 1]^d} B_i(u)du$ for $1 \leq i \leq k$, then
\begin{eqnarray*}
(P \alpha^*-M)^T (P \alpha^*-M) & \leq & \max_{1 \leq i \leq k} c_i \int_{[0, 1]^d} \sum_{1 \leq i \leq k} B_i(u)\big((\alpha^*)^T B(u) - c(u)\big)^2 du\\
& = & \Delta_1^2 O(1)/ \beta,
\end{eqnarray*}
which, together with the fact that $Pen(c, \alpha^*) = (d(d-1)/2)\Delta_1^2$,
implies that $\varepsilon_1 = \Delta_1^2 O(1) \rightarrow 0$ as $n \rightarrow \infty$.

To control $\varepsilon_2$,  note that from Lemma 2, we have 
\begin{eqnarray*}
E  \varepsilon_2^2 & \leq &  O\Big(\frac{k}{n}\Big)\Big(1+\Big(\frac{d}{h_{\min}}\Big)^2\Big)\Big(\frac{1}{h_{\min}}\Big)^d\\
& = & O\Big(\frac{k}{n}\Big)\Big(\frac{1}{h_{\min}}\Big)^{2+d},
\end{eqnarray*}
so $\varepsilon_2$ converges to 0 in probability. 

From the above discussion for $\varepsilon_1$ and $\varepsilon_2$,  it follows from (\ref{eq:3})
that $U$ converges to 0 in probability. From (\ref{eq:11}), 
\[
\|\hat{\alpha}^TB-c\|^2 \leq \Big(\frac{2}{\beta\min(eigen(P))}\Big) U^2 + 2 \Delta_1^2 = O(1) U^2 + 2 \Delta_1^2,
\]
so $\|\hat{c}-c\| = \|\hat{\alpha}^TB-c\|$ converges to 0 in probability.
\end{proof}

\end{document}